\documentclass[11pt]{article}

\usepackage{fullpage}

\usepackage[utf8]{inputenc}
\usepackage[english]{babel}

\usepackage{dsfont}
\usepackage{mathpazo}

\usepackage{titlesec}
\usepackage{authblk}
\usepackage{delimset}

\usepackage[hyphens]{url}
\usepackage[hypertexnames=false,colorlinks=true,linkcolor=blue,citecolor=black,urlcolor=black]{hyperref}
\usepackage{breakurl}

\usepackage{upgreek}
\usepackage{bbold}

\usepackage{amsmath, amsfonts, amssymb}
\usepackage{mathtools}
\usepackage{braket}

\usepackage{array}

\def\squareforqed{\hbox{\rlap{$\sqcap$}$\sqcup$}}
\def\qed{\ifmmode\squareforqed\else{\unskip\nobreak\hfil
\penalty50\hskip1em\null\nobreak\hfil\squareforqed
\parfillskip=0pt\finalhyphendemerits=0\endgraf}\fi}

\newtheorem{theorem}{Theorem}
\newtheorem{lemma}[theorem]{Lemma}
\newtheorem{fact}[theorem]{Fact}
\newtheorem{definition}[theorem]{Definition}
\newenvironment{proof}{\begin{trivlist}\item[]{\flushleft\bf Proof }}
{\qed\end{trivlist}}

\newcommand{\Cn}{\mathbb{C}}

\newcommand{\lqw}{\mathsf{A}_\mathrm{lazy}} % lackadaisical quantum walk
\newcommand{\vlqw}{\widehat{\mathsf{A}}_\mathrm{lazy}} % variant of l.q.w
\newcommand{\opw}{\mathsf{W}} % operator W
\newcommand{\opg}{\mathsf{G}} % operator G
\newcommand{\opshift}{\mathsf{S}_\mathrm{ff}} % Shift operator
\newcommand{\opswap}{\mathsf{SWAP}} % SWAP operator
\newcommand{\opc}{\mathsf{C}}
\newcommand{\opid}{\mathsf{I}}
\newcommand{\opp}{\mathsf{P}}
\newcommand{\oplw}{\widehat{\opp}}
\newcommand{\opref}{\textup{Ref}}
\newcommand{\iw}{\opp(s)}
\newcommand{\dm}[1]{\mathsf{D}(#1)}
\newcommand{\qiw}{\opw(\opp(s))} % quantum interpolated walk
\newcommand{\initlazy}{\ket{\overline{init}_\mathrm{lazy}}}
\newcommand{\initip}{\ket{\overline{init}_\mathrm{ip}}}
\newcommand{\initlip}{\ket{\overline{init}_{\widehat{\mathrm{ip}}}}}

\makeatletter
\newcommand{\cotqht}[2]{
  \def\buergi@arg{#1,#2}\@ifnextchar^\buergi@fexp{\mathsf{QHT}_\textnormal{cot}(#1, #2)}}
\def\buergi@fexp^#1{\mathsf{QHT}_\mathrm{cot}^{#1}(\buergi@arg)}
\makeatother
\newcommand{\lqiw}{\opw(\widehat\opp(s))}
\newcommand{\liw}{\widehat\opp(s)}
\newcommand{\sht}{\mathsf{HT}}
\newcommand{\eht}{\mathsf{HT}^{+}}
\newcommand{\intht}[1]{\mathsf{HT}_\mathrm{ip}(#1)}
\newcommand{\isolazy}{\widehat{\mathsf{T}}(s)}
\newcommand{\isops}{\mathsf{T}(s)}

\newcommand{\lazyphi}[2]{\ket{ {\widehat{\smash{\phi}\vphantom{i}}\vphantom{\phi}}_{#1}^{\hspace*{.5mm}\scriptstyle{#2}} }}

\newcommand{\nth}[1]{\ensuremath{{#1}^{\textup{th}}}}

\newcommand{\isoone}{\mathsf{R}_1}
\newcommand{\isotwo}{\mathsf{R}_2}
\newcommand{\isoe}{\mathsf{E}}

\newcommand{\basis}{e}
\newcommand{\nn}{N}

\renewcommand\bra[1]{{\langle{#1}|}}
\renewcommand\ket[1]{{|{#1}\rangle}}
\newcommand{\ketbra}[2]{\ket{#1}\!\bra{#2}}

\newcommand{\selfloop}{\circlearrowleft}

\DeclarePairedDelimiterX{\dnorm}[1]{\lVert}{\rVert}{#1}
\DeclarePairedDelimiterX{\dabs}[1]{\lvert}{\rvert}{#1}

\usepackage{tikz}
\usetikzlibrary{calc} 

\def\plusarrowlength{2.5pt}
\newcommand{\plusarrow}{\mathbin{\text{$
      \begin{tikzpicture}[
        baseline,
        ]
        \node[
        inner xsep=-1pt,
        inner ysep=-.5pt,
        outer sep=0pt,
        anchor=base
        ] (plus) {\phantom{+}};
        \begin{scope}[to-to]
          \draw ($(plus.west)+(-\plusarrowlength,0)$) -- ($(plus.east)+(\plusarrowlength,0)$);
          \draw ($(plus.north)+(0,\plusarrowlength)$) -- ($(plus.south)+(0,-\plusarrowlength)$);
        \end{scope}
      \end{tikzpicture}$}}}

\title{Analysis of Lackadaisical Quantum Walks\footnote{
Published in:
Quantum Information and Computation, Vol.~20, No.~13--14 (2020) 1137--1152,
by Rinton Press.\newline
\url{http://www.rintonpress.com/xxqic20/qic-20-1314/1137-1152.pdf}}}

\author{Peter H{\o}yer}
\author{Zhan Yu}
\affil{Department of Computer Science, University of Calgary, Canada}

\date{}

\begin{document}
\maketitle

\begin{abstract} 
  The lackadaisical quantum walk is a quantum analogue of the lazy random walk
  obtained by adding a self-loop to each vertex in the graph.  We analytically
  prove that lackadaisical quantum walks can find a unique marked vertex on
  any regular locally arc-transitive graph with constant success probability
  quadratically faster than the hitting time.  This result proves several
  speculations and numerical findings in previous work, including the
  conjectures that the lackadaisical quantum walk finds a unique marked vertex
  with constant success probability on the torus, cycle, Johnson graphs, and
  other classes of vertex-transitive graphs.  Our proof establishes and uses a
  relationship between lackadaisical quantum walks and quantum interpolated
  walks for any locally arc-transitive graph.
\end{abstract}

% ==============================================================================

\section{Introduction}
\label{sec_intro}
Searching is one of the most important tasks in computer science, and
searching algorithms have been well studied from both classical and quantum
aspects. One of the most famous quantum algorithms, Grover's search
algorithm~\cite{Gro96}, can search an $N$-item unstructured database in
$O(\sqrt{N})$ steps, which is quadratically faster than classical
searching. The method used in Grover's algorithm is generalized as amplitude
amplification in~\cite{BHMT02}.

Searching structured databases can be modeled as spatial search problems on
graphs, where the vertices of the graph represent the search space. A subset
of the vertices are marked, and the goal is to find one of the marked
vertices. One classical strategy is to use a random walk to traverse the graph
along its edges until a marked vertex is reached. The expected number of steps
$\sht$ required to reach a marked vertex by a random walk is called the
\emph{hitting time}. Quantum walks, which are quantum counterparts of random
walks, are used to develop quantum algorithms for spatial search problems.

Szegedy introduced a generic method of constructing a quantum walk from a
reversible random walk~\cite{Sze04}. The resulting quantum walk uses
$O(\sqrt{\sht})$ steps, which yields a quadratic speedup over the random
walk. Szegedy's algorithm does not necessarily find a marked vertex, but it
can detect the presence of a marked vertex. Krovi et al.~\cite{KMOR16} later
proposed a quantum algorithm based on the novel idea of interpolated
walks. They applied Szegedy's correspondence on the interpolated walks and
called the resulting algorithms quantum interpolated walks.  Quantum
interpolated walks can find a marked vertex in $O(\sqrt{\eht})$ steps for any
reversible random walk, where $\eht$ is the \emph{extended hitting time} of
the random walk. When there is a unique marked vertex, then $\eht = \sht$ and
this quantum walk thus achieves a quadratic speedup over the random walk. When
there are multiple marked vertices, $\eht$ may be asymptotically larger than
$\sht$~\cite{AK15}. Dohotaru and H{\o}yer~\cite{DH17b} achieved the same
result by introducing a different framework called controlled quantum walks.

Quantum walks are commonly applied on graphs without self-loops.  The
lackadaisical quantum walk proposed by Wong~\cite{Won15b}, is a quantum
analogue of the lazy random walk, which adds a self-loop of weight $\ell$ to
each vertex.  The lackadaisical quantum walk generalizes the three-state lazy
quantum walk on the line~\cite{Won15b,IKS05}.  The idea of adding self-loops
was first applied by Ambainis et al.~\cite{AKR05}, who showed that a quantum
walk on a complete graph with a self-loop on every vertex corresponds to
Grover's algorithm.  The lackadaisical quantum walk can also be viewed as a
coined quantum walk, and Wong uses this to demonstrate that the asymptotic
behavior of some coined quantum walks can be improved by modifying the
coin~\cite{Won18b}.

The lackadaisical quantum walk, with a unique marked vertex, has been studied
on several classes of graphs.  First, the complete graph was studied by
Wong~\cite{Won15b,Won18a,Won17a}, who proved analytically that the
lackadaisical quantum walk with $\ell = 1$ finds a unique marked vertex in
$O(\sqrt{N})$ steps with probability close to~1.

Second, the $\sqrt{N} \times \sqrt{N}$ torus was independently studied by
Wong~\cite{Won18b} and Wang et al.~\cite{WZWY17arxiv}, who both showed
numerically that the lackadaisical quantum walk finds a unique marked vertex
with probability close to~1.  The value of $\ell$ is $\frac{4}{N}$ in
Wong~\cite{Won18b} and $\frac{4}{N-1}$ in Wang et al.~\cite{WZWY17arxiv}.  The
result is strongly supported by the experiments, but with no analytical proofs
of the complexity and success probability, the result is stated as a
conjecture.

Third, the cycle was studied by Giri and Korepin~\cite{GK19}, who showed
numerically that by setting $\ell = \frac{2}{N}$, the lackadaisical quantum
walk finds a unique marked vertex with probability at least a constant.

Fourth, regular complete bipartite graphs were studied by Rhodes and
Wong~\cite{RW19b}, who proved analytically that the lackadaisical quantum walk
with $\ell = \frac{1}{2}$ finds a unique marked vertex in $O(\sqrt{N})$ steps
with probability close to~1.

Fifth, in a recent paper, Rhodes and Wong~\cite{RW20} study a collection of
graphs.  Their collection is a rich sample of vertex-transitive graphs,
comprised of the following instances and classes of graphs:
arbitrary-dimensional cubic lattices, Paley graphs, the two Latin square
graphs with strongly regular parameters $(9, 6, 3, 6)$ and
$(1024, 93, 32, 6)$, triangular graphs, Johnson graphs, and the hypercube.
They show numerically that by setting $\ell = \frac{d}{N}$, the lackadaisical
quantum walk finds a unique marked vertex with probability at least a
constant.  Here $d$ is the degree of the vertices.  They propose that this
holds for all vertex-transitive graphs with a unique marked vertex, and they
propose that the weight of self-loop $\ell = \frac{d}{N}$ optimally boosts the
success probability.

In this work, we prove analytically that the lackadaisical quantum walk finds
a unique marked vertex with probability at least a constant on all of the
above-mentioned graphs when choosing the weight $\ell$ of the self-loops as
listed above.  More generally, we prove that for any $d$-regular locally
arc-transitive graph, by adding a self-loop of weight $\ell = \frac{d}{N}$ on
each vertex, the lackadaisical quantum walk finds a unique marked vertex with
probability at least a constant.

Our main results are stated as Theorem~\ref{thm1} and Theorem~\ref{thm2} in
Section~\ref{sec_main}.  Theorem~\ref{thm1} states that the quantum hitting
time of lackadaisical quantum walks and quantum interpolated walks are of the
same order.  Theorem~\ref{thm2} states that the $\ell_2$-distance between the
two quantum states of the lackadaisical quantum walk and the quantum
interpolated walk, respectively, remains negligible for any number of steps
that is in the order of the quantum hitting time.  The two theorems hold for
any regular locally arc-transitive graph.

In Section~\ref{proof of thm1}, we prove Theorem~\ref{thm1} by introducing a
variant of lackadaisical quantum walks as an intermediate walk operator and
then giving an exact relationship between the quantum hitting times of all
three quantum walk operators.  In Section~\ref{proof of thm2}, we construct
isometries and use them to upper bound the $\ell_2$-distance between the
resulting states of the intermediate quantum walk operator and the quantum
interpolated walk after any number of steps.

By combining the two main Theorems with the analysis of quantum interpolated
walks given in~\cite{KMOR16} and the analysis of controlled quantum walks
given in~\cite{DH17b}, we complete the analytical proofs of the complexity and
success probability of lackadaisical walks on regular locally arc-transitive
graphs.  We also prove that self-loops of weight $\ell = \frac{d}{N}$
correspond to an interpolation value of $s=1-\frac{1}{N}$, which closely
matches the value of $s=1-\frac{1}{N-1}$ used in~\cite{KMOR16}.

The main technical contribution in our work is the use of locally
arc-transitivity to establish a connection between lackadaisical quantum walks
and quantum interpolated walks.  The definition of locally arc-transitivity is
given in Section~\ref{sec_main}.  We discuss the relationship between locally
arc-transitivity, vertex transitivity and other graph properties in
Section~\ref{sec_arc}, and we conclude in Section~\ref{sec_conclusion}.

Our results are for the case when there is a unique marked vertex.  When there
are multiple marked vertices, the lackadaisical quantum walk may fail in
finding a marked vertex.  Nahimovs~\cite{Nah19} proves that, on a
$\sqrt{N} \times \sqrt{N}$ torus with \emph{two} marked vertices placed
adjacently to each other, the lackadaisical quantum walk has a stationary
state that is close to the initial state, which implies that the walk finds a
marked vertex with probability no bigger than~$O(1/N)$.

\section{Two Quantum Walks}\label{sec_twowalks}
The graphs that we apply the quantum walks on, are regular undirected
connected graphs with a unique marked vertex.  A graph is said to be
\emph{$d$-regular} if every vertex has degree~$d$.  We will interchangeably
consider an undirected graph as a directed graph, where we consider each edge
$\{x,y\}$ between two distinct vertices as two arcs $(x,y)$ and $(y,x)$.  When
introducing self-loops below, we also interchangeably consider each edge
$\{x,x\}$ as a single arc $(x,x)$.  For each vertex~$x$ in turn, fix any
ordering of the $d$ neighbors $y_1, y_2, \ldots, y_d$ of~$x$.  We refer to
$y_i$ as the \nth{i} neighbor of~$x$, and the arc $(x,y_i)$ as the \nth{i}
outgoing arc of~$x$.  

Let $N$ denote the number of vertices, and let $\mathcal{H}_N$ be the Hilbert
space spanned by the vertices of the graph.  To each vertex we associate a
coin register in the Hilbert space $\mathcal{H}_d$ spanned by the basis
$\{\ket{\basis_1},\ket{\basis_2}, \ldots, \ket{\basis_d}\}$.  The coined
quantum walk takes place in the Hilbert space
$\mathcal{H}_N \otimes \mathcal{H}_d$, in which the state
$\ket{x}\ket{\basis_i}$ represents the arc from $x$ to its \nth{i}
neighbor~$y_i$.

\vspace*{12pt}
\begin{definition}[Lackadaisical quantum walks~\cite{Won15b}]\label{lqw}
  Given a $d$-regular graph with a unique marked vertex $m$, by adding a
  self-loop of weight $\ell$ to every vertex, the coined Hilbert space
  $\mathcal{H}_{d+1}$ is spanned by
  $\bigl\{\ket{\basis_1},\ket{\basis_2}, \ldots, \ket{\basis_d},
  \ket{\circlearrowleft}\bigr\}$. The lackadaisical quantum walk is defined as
  \begin{equation}
    \lqw = \opw \cdot \opg.
  \end{equation}
  Here $\opw$ is the quantum walk operator (without searching) defined as
  $\opw = \opshift\cdot(\opid_N \otimes \opc)$,
  where operator
  $\opc = 2 \ketbra{c}{c} - \opid_{d+1}$
  with
  \begin{equation}\label{def_c}
    \ket{c} = \frac{1}{\sqrt{d+\ell}} \bigl(\ket{\basis_1} + \ket{\basis_2} + \cdots + \ket{\basis_d} + \sqrt{\ell} \ket{\circlearrowleft}\bigr)
  \end{equation}
  is the diffusion coin for a weighted graph and $\opshift$ is the flip-flop
  shift operator~\cite{AKR05} defined as
  \begin{equation*}
    \opshift\colon
    \left\{
      \begin{aligned}
        &\ket{x, \basis_i} &&\mapsto&& \ket{y, \basis_j} \\
        &\ket{x, \circlearrowleft} &&\mapsto&& \ket{x, \circlearrowleft},
      \end{aligned}
    \right.
  \end{equation*}
  where $y$ is the \nth{i} neighbor of~$x$, and $x$ is the \nth{j} neighbor
  of~$y$.  A~query to the oracle is defined as
  \begin{equation}\label{def_G}
    \opg = (\opid_N - 2\ketbra{m}{m}) \otimes \opid_{d+1},
  \end{equation}
  where $\ket{m}$ denotes the unique marked vertex. 
\end{definition}

The lackadaisical quantum
walk $\lqw$ begins in the state
\begin{equation*}
  \initlazy = \frac{1}{\sqrt{N-1}}\sum_{x\neq m}\ket{x} \otimes \ket{c},
\end{equation*}
which is a uniform superposition over all unmarked vertices.

Given a graph $G$, define the \emph{random walk} $\opp = \opp(G)$, where
$\opp_{xy}$ is the transition probability from vertex $x$ to vertex $y$.  If
vertex $y$ is a neighbor of vertex $x$, then $\opp_{xy} = \frac{1}{\deg(x)}$,
where $\deg(x)$ denotes the degree of~$x$, and otherwise $\opp_{xy} = 0$. The
stationary distribution of $\opp$ is denoted by $\uppi$, and $\uppi_v$ denotes
the probability of being in vertex~$v$ in the stationary distribution.  The
\emph{absorbing walk} $\opp'$ is obtained from $\opp$ by replacing all
outgoing transitions from any marked vertex with self-loops, that is,
$\opp'_{xy} = \opp_{xy}$ for all unmarked vertices $x$ and all $y$, and, for
any marked vertex~$m$, $\opp'_{mm} = 1$ and $\opp'_{my} = 0$ for $y \neq m$.

Given $0 \leq s \leq 1$, the \emph{interpolated walk} $\iw$ is defined as
\begin{equation*}
  \iw = (1-s)\opp + s\opp'.
\end{equation*}
Note that $\iw_{xy}$ is the transition probability from vertex $x$ to vertex
$y$ in $\iw$.

\begin{definition}[Quantum interpolated walks~\cite{KMOR16}]
  Applying Szegedy's correspondence~\cite{Sze04} on the interpolated walk
  $\iw$, we construct the quantum interpolated walk
  \begin{equation*}
    \qiw = \opswap \cdot \opref(\mathcal{A}).
  \end{equation*}
  The operator $\opref(\mathcal{A})$ is a reflection about the
  subspace $\mathcal{A}$ spanned by $\{\ket{x, \iw_x}\}$ for all vertices
  $x$, where
  \begin{equation*}
    \ket{\iw_x} = \sum_{y} \sqrt{\iw_{xy}} \ket{y}
  \end{equation*}
  is a superposition over the neighbors of~$x$. The operator $\opswap$ swaps
  the two registers. 
\end{definition}

The initial state for the quantum interpolated walk $\qiw$ is
\begin{equation*}
  \initip = \frac{1}{\sqrt{N-1}} \sum_{x\neq m} \ket{x}\otimes \ket{\iw_{x}}.
\end{equation*}

\begin{definition}[The cotangent quantum hitting time \cite{DH17b}]
  \label{cotangent qht}
  The \emph{cotangent quantum hitting time} of a quantum walk $\mathsf{U}$ on
  a state $\ket{w}$ is
  \begin{equation*}
    \cotqht{\mathsf{U}}{\ket{w}} = \sqrt{\sum_{\phi_k^\pm\neq 1} \dabs[\big]{
        \braket{\phi_k^\pm | w} }^2 \cot^2{\frac{\theta_k}{2}}}
  \end{equation*}
  where $\ket{\phi_k^\pm}$ are the eigenvectors of $\mathsf{U}$ corresponding
  to the eigenvalues $\phi_k^\pm = e^{\pm i\theta_k}$.
\end{definition}

\section{Main Theorems}\label{sec_main}

Rhodes and Wong~\cite{RW20} and the other earlier work consider lackadaisical
quantum walks on certain instances and classes of regular vertex-transitive
graphs.  One common property of these instances and classes of graphs is that
they are locally arc-transitive.  A~graph $G$ is said to be \emph{locally
  arc-transitive} if for any vertex $u$ with neighbors $v_1$ and $v_2$, there
exists an automorphism $\sigma$ of~$G$ that maps the arc $(u,v_1)$ to the arc
$(u,v_2)$.  That is, there is an automorphism that fixes $u$ while mapping any
one of $u$'s neighbors to any other of $u$'s neighbors.

Any connected \emph{locally arc-transitive} graph must be biregular, since for
any two vertices $u$ and~$v$ connected by a path of even length, there exists
an automorphism that maps $u$ to~$v$.  By the same argument, if the graph
contains an odd cycle, then it must be regular.  For simplicity, in this
paper, we consider only locally arc-transitive graphs that are regular.

We prove that the lackadaisical quantum walk $\lqw$ searches regular locally
arc-transitive graphs for a unique marked vertex in $O(\sqrt{\sht})$ steps
with constant success probability.

\begin{theorem}\label{thm0}
  Let $G$ be a $d$-regular locally arc-transitive graph with $N$ vertices and
  a unique marked vertex~$m$.  The lackadaisical quantum walk with selfloop of
  weight $\ell=\frac{d}{N}$ can find $m$ with constant success probability in
  $O\bigl(\sqrt{\sht(\opp(G), \{m\})}\bigr)$ steps.
\end{theorem}

Here $\sht(\opp(G), \{m\})$ denotes the hitting time on~$G$ when the unique
marked vertex is~$m$.  Theorem~\ref{thm0} follows from Theorems~\ref{thm1}
and~\ref{thm2}.  Theorem~\ref{thm1} shows that the quantum hitting time of
lackadaisical quantum walks is of the same order as the quantum hitting time
of quantum interpolated walks. Theorem~\ref{thm2} then shows that the
$\ell_2$-distance between the resulting states of $\lqw$ and $\qiw$ is small
for any number of steps in $O(\sqrt{\sht})$.

\begin{theorem}\label{thm1}
  Consider a $d$-regular locally arc-transitive graph with $N$ vertices and a
  unique marked vertex~$m$. For the lackadaisical quantum walk $\lqw$, add
  self-loops of weight $\ell=\frac{d}{N}$ on every vertex. For the quantum
  interpolated walk $\qiw$, choose $s = 1- \frac{\ell}{d}$. Then
  \begin{equation*}
    \cotqht{\lqw}{\initlazy}^2 = \frac{N+1}{N}\cotqht{\qiw}{\initip}^2 + \frac{1}{2N-1},
  \end{equation*}
  and
  \begin{equation*}
    \cotqht{\qiw}{\initip} \in O\Bigl(\sqrt{\sht(\opp, \{m\})}\Bigr).
  \end{equation*}
\end{theorem}

\begin{theorem}\label{thm2}
  Set $T_0 =\big\lfloor c \cdot \sqrt{\sht(\opp, \{m\})} \big\rfloor$ for any
  fixed constant $c \geq 1$. For all $t \leq T_0$,
  \begin{equation*}
    \dnorm[\big]{\lqw^t \initlazy - \qiw^t \initip}_2
    \in O\biggl(\frac{1}{N^{1/4}}\biggr).
  \end{equation*}
\end{theorem}

Our main results show that the lackadaisical quantum walk $\lqw$ is closely
related to the quantum interpolated walk $\qiw$. This relationship permits us
to analyze the quantum hitting time and behavior of the lackadaisical quantum
walk $\lqw$ using known results about quantum interpolated walks. It is shown
in~\cite{KMOR16} that $\qiw$ finds a unique marked element in
$O(\cotqht{\qiw}{\initip})$ steps with constant success probability, where we
use the tight bounds on the cotangent quantum hitting time given in Appendix~A
in~\cite{DH17b}.  This proves the conjectures and numerical findings
in~\cite{RW20} and the other earlier work on the complexity and success
probability of lackadaisical quantum walks on those graphs.  We prove
Theorem~\ref{thm1} in Sec.~\ref{proof of thm1} and Theorem~\ref{thm2} in
Sec.~\ref{proof of thm2}.

Throughout the remaining sections, we fix $\ell = \frac{d}{N}$ for the
lackadaisical quantum walk and $s = 1- \frac{\ell}{d}$ for the quantum
interpolated walk.

\section{Technical Preliminaries}\label{sec_tech}

Define the lazy random walk on $G$ as
\begin{equation}\label{lazy walk}
  \oplw = \frac{d}{d+\ell}\cdot \opp + \frac{\ell}{d+\ell} \cdot \opid_{N},
\end{equation}
obtained by adding a self-loop of weight $\ell$ to every vertex.  The
interpolation of a lazy random walk is then denoted
\begin{equation}\label{interpolated lazy walk}
  \liw = (1-s) \cdot \oplw + s \cdot \oplw',
\end{equation}
where $\oplw' = (\oplw)'$ is the absorbing walk derived from the lazy random
walk~$\oplw$.

We apply Szegedy's correspondence on $\iw$ and $\liw$. For convenience, we
only show the details on constructing $\qiw$. Applying Szegedy's
correspondence on $\liw$ is similar, except we use `$\ \widehat{\ }\ $' when
referring to $\liw$. The discriminant~\cite{Sze04} of the interpolated walk
$\iw$ is
\begin{equation*}
  \dm{\iw} = \sqrt{\iw \circ \iw^T },
\end{equation*}
where the Hadamard product ``$\circ$'' and the square root are taken
entry-wise, and the $T$ denotes matrix transposition.  We denote the
corresponding eigenvalues of $\dm{\iw}$ by $\lambda_k$, where
$k=1,\ldots,\nn$.  Let $-\pi/2 \leq \theta_k \leq \pi/2$ be angles so that
$\lambda_k = \cos\theta_k$.

The \emph{interpolated hitting time} \cite{KMOR16} of an interpolated walk $\iw$ is
\begin{equation}\label{ip hitting time}
  \intht{\iw} = \sum_{\lambda_k \neq 1}
  \frac{\lvert\braket{\lambda_k | \sqrt{\bar{\uppi}}}\rvert^2}{1-\lambda_k},
\end{equation}
where $\ket{\lambda_k}$ are the corresponding eigenvectors and $\bar{\uppi}$
is the uniform distribution over all unmarked vertices.

To analyze the quantum analogue $\qiw$ of the interpolated walk $\iw$, define
the isometry
\begin{equation*}
  \isops = \sum_x \ketbra{x, \iw_x}{x}.
\end{equation*}
The quantum walk $\qiw$ has a unique eigenvector
$\ket{\phi_{\nn}} = \isops\ket{\lambda_{\nn}}$ with eigenvalue $\phi_{\nn}= 1$. The
remaining $2(\nn-1)$ eigenvalues and eigenvectors are
\begin{align*}
  \phi^\pm_k &= e^{\pm i \theta_k},
  & \ket{\phi^\pm_k} &= \frac{\isops\ket{\lambda_k} \pm i(\isops\ket{\lambda_k})^\perp}{\sqrt{2}}
\end{align*}
for $k = 1, \ldots, \nn-1$.  The phases of the eigenvectors can be chosen so
that they satisfy that
\begin{equation*}
  \isops \ket{\lambda_k} = \frac{1}{\sqrt{2}}(\ket{\phi_k^+} + \ket{\phi_k^-}).
\end{equation*}
We decompose $\sqrt{\bar{\uppi}}$ into the basis of $\dm{\iw}$ for scalars
$\alpha_k$,
\begin{equation}\label{decompose pibar}
  \sqrt{\bar{\uppi}} = \sum_{k=1}^{\nn} \alpha_k \ket{\lambda_k},
\end{equation}
and write the initial state as
\begin{equation*}
  \initip = \isops \sqrt{\bar{\uppi}} = \alpha_{\nn} \ket{\phi_{\nn}} + \frac{1}{\sqrt{2}} \sum_{k=1}^{\nn-1} \alpha_k (\ket{\phi_k^+} + \ket{\phi_k^-}).
\end{equation*}
Applying the quantum walk $\qiw$ for $t$ times on $\initip$, yields the state
\begin{equation}\label{decompose state}
  \qiw^t \initip = \alpha_{\nn} \ket{\phi_{\nn}} + \frac{1}{\sqrt{2}} \sum_{k=1}^{\nn-1} \alpha_k \bigl((e^{i\theta_k})^t\ket{\phi_k^+} + (e^{-i\theta_k})^t\ket{\phi_k^-}\bigr).
\end{equation}

\section{Proof of Theorem~\ref{thm1}}\label{proof of thm1}
To prove Theorem~\ref{thm1}, we use a variant of lackadaisical quantum walk as
an intermediate quantum walk operator. The lackadaisical quantum walk $\lqw$
in Definition~\ref{lqw} uses a query to the oracle~$\opg$. We define a query
to a different oracle as
\begin{equation}\label{def_Ghat}
  \widehat\opg = \opid_{(d+1)N} - 2 \big(
  \ketbra{m, \circlearrowleft}{m, \circlearrowleft}
  + \ketbra{m, \plusarrow}{m, \plusarrow}\big),
\end{equation}
where
\begin{equation*}
  \ket{\plusarrow} = \frac{1}{\sqrt{d}} \sum_{i=1}^{d} \ket{\basis_i}
\end{equation*}
denotes an equally weighted superposition over all the $d$ outgoing arcs of
any vertex.  Using the query $\widehat\opg$, we define the following variant
of lackadaisical quantum walks,
\begin{equation*}
  \vlqw = \opw \cdot \widehat\opg.
\end{equation*}
We first show that for any locally arc-transitive graph, we can replace the
query $\opg$ by the modified query~$\widehat\opg$ without altering the
evolution of the walk.

Let $G$ be a $d$-regular locally arc-transitive graph with a unique marked
vertex~$m$.  Write the state of the system after $t$ steps of the walk~$\lqw$
on the initial state $\initlazy$,
\begin{equation}
  \label{eq_stateafterdsteps}
  \lqw^t \initlazy 
  = \bigg(\sum_{u \neq m; \,i} \alpha_{u,i} \ket{u,\basis_i} \bigg)
  + \bigg(\sum_{i} \alpha_{m,i} \ket{m,\basis_i} \bigg)
  + \bigg(\sum_{v \in V(G)} \alpha_{v,\circlearrowleft} \ket{v,\circlearrowleft} \bigg),
\end{equation}
for some amplitudes~$\alpha$.  We first show that locally arc-transitivity
implies that the amplitudes of the outgoing arcs of the marked state $m$
remain equal after any number of iterations.

\begin{lemma}\label{lem_invar}
  For all $t \geq 0$, $\alpha_{m,i} = \alpha_{m,j}$ for all outgoing arcs
  $(m,y_i)$ and $(m,y_j)$ of the marked vertex~$m$.
\end{lemma}

\begin{proof}
  We define an application of an automorphism $\sigma$ of $G$ on quantum
  states and operators acting on $\mathcal{H}_N \otimes \mathcal{H}_{d+1}$ as
  follows.  For any vertex $u$ and its \nth{i} outgoing arc, let
  $\sigma (\ket{u,\basis_i}) = \ket{u',\basis_j}$, where $\sigma(u) = u'$, $v$
  is the endpoint of the $\nth{i}$ arc from~$u$, $\sigma(v) = v'$, and $v'$ is
  the endpoint of the $\nth{j}$ arc from~$u$.  For any vertex $u$ and its
  self-loop, let $\sigma (\ket{u,\selfloop}) = \ket{u',\selfloop}$.
  Generalize and define the action on the adjoint as
  $\sigma(\bra{u, \basis_i}) = \bra{u', \basis_j}$ and
  $\sigma(\bra{u, \selfloop}) = \bra{u', \selfloop}$, and extend to operators
  and the entire Hilbert space by composition and linearity.

  For any automorphism $\sigma$, then $\sigma(\opshift) = \opshift$, since
  $\opshift$ changes the direction of every arc.  Similarly,
  $\sigma(\opid_N \otimes \opc) = (\opid_N \otimes \opc)$, since an
  automorphism preserves the neighborhood of any vertex.
  
  Let $v_1$ and $v_2$ be the endpoints of any two distinct outgoing arcs
  of~$m$ (excluding the self-loop).  Since $G$ is locally arc-transitive,
  there exists an automorphism $\sigma_m$ that fixes $m$ and that maps the arc
  $(m,v_1)$ to $(m,v_2)$.  Let $\Delta_m$ be a set of $d(d-1)$ such
  automorphisms, one for each pair $v_1,v_2$ of distinct neighbors of~$m$.
  For any automorphism $\sigma_m \in \Delta_m$ we have that
  $\sigma_m(\opg) = \opg$ and $\sigma_m(\initlazy) = \initlazy$.

  The proof of the lemma now follows readily by mathematical induction on~$t$.
  Assume that $\sigma_m(\lqw^{t} \initlazy) = \lqw^{t} \initlazy$ for any
  $\sigma_m \in \Delta_m$ immediately prior to the \nth{(t+1)} application
  of~$\lqw$.  Then by the above arguments, it holds immediately after the
  \nth{(t+1)} application of~$\lqw$.  Since
  $\sigma_m(\lqw^{t+1} \initlazy) = \lqw^{t+1} \initlazy$, then
  $\sigma_m(\sum_{i} \alpha_{m,i} \ket{m,\basis_i}) = \sum_{i} \alpha_{m,i}
  \ket{m,\basis_i}$, and thus $\alpha_{m,i} = \alpha_{m,j}$ for all outgoing
  arcs $(m,y_i)$ and $(m,y_j)$.
\end{proof}

\begin{lemma}\label{lem1}
  For all $t\geq 0$,
  \begin{equation*}
    \lqw^t \initlazy = \vlqw^t \initlazy.
  \end{equation*}
\end{lemma}

\begin{proof}
  The difference between $\lqw$ and $\vlqw$ is that $\lqw$ uses a query to the
  oracle $\opg$ in Eq.~(\ref{def_G}) and $\vlqw$ uses a different query operator
  $\widehat{\opg}$ given by Eq.~(\ref{def_Ghat}).  By Lemma~\ref{lem_invar},
  whenever we apply ${\opg}$ or $\widehat{\opg}$, we have a state where the
  second summand in Eq.~(\ref{eq_stateafterdsteps}) can be written in the form
  $\alpha_{m,\plusarrow} \ket{m,\plusarrow} + \alpha_{m,\selfloop}
  \ket{m,\selfloop}$.  On such states, ${\opg}$ and $\widehat{\opg}$ act
  identically.
\end{proof}

Consider the two quantum walks $\vlqw$ and $\lqiw$. The search space of
$\vlqw$ is the Hilbert space $\mathcal{H}_N \otimes \mathcal{H}_{d+1}$. The
search space of $\lqiw$ is the Hilbert space
$\mathcal{H}_N \otimes \mathcal{H}_N$. By Szegedy's correspondence, the
quantum interpolated walk $\lqiw$ takes place within a smaller subspace
$\Cn^{(d+1)N}$ of the full Hilbert space
$\mathcal{H}_N \otimes \mathcal{H}_N$. We identify the subspace $\Cn^{(d+1)N}$
with $\mathcal{H}_N \otimes \mathcal{H}_c$ by defining an isometry
$\isoe\colon \Cn^{(d+1)N} \rightarrow \mathcal{H}_N \otimes \mathcal{H}_{d+1}$
as follows.

For all vertices $x$ in the $d$-regular arc-transitive graph, and all
neighbors $y_i$ of $x$, let
\begin{equation*}
  \isoe \ket{x, y_i} = \ket{x, e_i}
\end{equation*}
and let
\begin{equation*}
  \isoe \ket{x, x} = 
  \begin{cases}
    \hphantom{-}\ket{x, \circlearrowleft} \hphantom{{}-{}} 
    & \textup{if $x$ is unmarked}\\
    -\ket{x, \circlearrowleft} 
    & \textup{if $x$ is marked.}
  \end{cases}
\end{equation*}

\begin{lemma}\label{lem2}
  \begin{equation*}
    \vlqw = \isoe\cdot \lqiw \cdot \isoe^\dagger.
  \end{equation*}
\end{lemma}

\begin{proof}
  The quantum circuit of the quantum walk $\vlqw = \opw \cdot \widehat\opg$ is
  given in Figure~\ref{circForHlqw}.

  \begin{figure}[htbp]
    \centering \begin{tikzpicture}[scale=1.500000,x=1pt,y=1pt]
\filldraw[color=white] (0.000000, -8.500000) rectangle (182.000000, 20.500000);
% Drawing wires
% Line 3: a W breadth=40 size=40
\draw[color=black] (0.000000,28.500000) -- (182.000000,28.500000);
% Line 4: b W
\draw[color=black] (0.000000,0.000000) -- (182.000000,0.000000);
% Done with wires; drawing gates
% Line 5: a P $m$ b G:width=40 $\mathsf{Ref}(\ket{\circlearrowleft})$
\draw (26.000000,28.500000) -- (26.000000,0.000000);
\begin{scope}
\draw[fill=white] (26.000000, 28.500000) circle(7.000000pt);
\clip (26.000000, 28.500000) circle(7.000000pt);
\draw (26.000000, 28.500000) node {$m$};
\end{scope}
\begin{scope}
\draw[fill=white] (26.000000, -0.000000) +(-45.000000:28.284271pt and 9.899495pt) -- +(45.000000:28.284271pt and 9.899495pt) -- +(135.000000:28.284271pt and 9.899495pt) -- +(225.000000:28.284271pt and 9.899495pt) -- cycle;
\clip (26.000000, -0.000000) +(-45.000000:28.284271pt and 9.899495pt) -- +(45.000000:28.284271pt and 9.899495pt) -- +(135.000000:28.284271pt and 9.899495pt) -- +(225.000000:28.284271pt and 9.899495pt) -- cycle;
\draw (26.000000, -0.000000) node {$\mathsf{Ref}(\ket{\circlearrowleft})$};
\end{scope}
% Line 6: a P $m$ b G:width=40 $\mathsf{Ref}(\ket{\plusarrow})$
\draw (78.000000,28.500000) -- (78.000000,0.000000);
\begin{scope}
\draw[fill=white] (78.000000, 28.500000) circle(7.000000pt);
\clip (78.000000, 28.500000) circle(7.000000pt);
\draw (78.000000, 28.500000) node {$m$};
\end{scope}
\begin{scope}
\draw[fill=white] (78.000000, -0.000000) +(-45.000000:28.284271pt and 9.899495pt) -- +(45.000000:28.284271pt and 9.899495pt) -- +(135.000000:28.284271pt and 9.899495pt) -- +(225.000000:28.284271pt and 9.899495pt) -- cycle;
\clip (78.000000, -0.000000) +(-45.000000:28.284271pt and 9.899495pt) -- +(45.000000:28.284271pt and 9.899495pt) -- +(135.000000:28.284271pt and 9.899495pt) -- +(225.000000:28.284271pt and 9.899495pt) -- cycle;
\draw (78.000000, -0.000000) node {$\mathsf{Ref}(\ket{\plusarrow})$};
\end{scope}
% Line 7: a P $x$ b G:width=40 $\mathsf{Ref}(\ket{c})$
\draw (130.000000,28.500000) -- (130.000000,0.000000);
\begin{scope}
\draw[fill=white] (130.000000, 28.500000) circle(7.000000pt);
\clip (130.000000, 28.500000) circle(7.000000pt);
\draw (130.000000, 28.500000) node {$x$};
\end{scope}
\begin{scope}
\draw[fill=white] (130.000000, -0.000000) +(-45.000000:28.284271pt and 9.899495pt) -- +(45.000000:28.284271pt and 9.899495pt) -- +(135.000000:28.284271pt and 9.899495pt) -- +(225.000000:28.284271pt and 9.899495pt) -- cycle;
\clip (130.000000, -0.000000) +(-45.000000:28.284271pt and 9.899495pt) -- +(45.000000:28.284271pt and 9.899495pt) -- +(135.000000:28.284271pt and 9.899495pt) -- +(225.000000:28.284271pt and 9.899495pt) -- cycle;
\draw (130.000000, -0.000000) node {$\mathsf{Ref}(\ket{c})$};
\end{scope}
% Line 8: a b G $\opshift$
\draw (169.000000,28.500000) -- (169.000000,0.000000);
\begin{scope}
\draw[fill=white] (169.000000, 14.250000) +(-45.000000:9.899495pt and 30.052038pt) -- +(45.000000:9.899495pt and 30.052038pt) -- +(135.000000:9.899495pt and 30.052038pt) -- +(225.000000:9.899495pt and 30.052038pt) -- cycle;
\clip (169.000000, 14.250000) +(-45.000000:9.899495pt and 30.052038pt) -- +(45.000000:9.899495pt and 30.052038pt) -- +(135.000000:9.899495pt and 30.052038pt) -- +(225.000000:9.899495pt and 30.052038pt) -- cycle;
\draw (169.000000, 14.250000) node {$\opshift$};
\end{scope}
% Done with gates; drawing ending labels
% Done with ending labels; drawing cut lines and comments
% Done with comments
\end{tikzpicture}
    \caption{Quantum circuit of $\vlqw$.}
    \label{circForHlqw}
  \end{figure}

  The two states $\ket{\circlearrowleft}$ and $\ket{\plusarrow}$ are
  orthogonal, and they span a two-dimensional subspace of the coin
  space~$\mathcal{H}_c$. The coin state
  $\ket{c} = \sqrt{\frac{d}{d+\ell}} \ket{\plusarrow} +
  \sqrt{\frac{\ell}{d+\ell}} \ket{\circlearrowleft}$ is in this
  two-dimensional subspace. Let
  $\ket{c^\perp} = \sqrt{\frac{\ell}{d+\ell}} \ket{\plusarrow} -
  \sqrt{\frac{d}{d+\ell}} \ket{\circlearrowleft}$ be the state that is
  orthogonal to the coin state~$\ket{c}$ in this two-dimensional subspace.  In
  Figure~\ref{circForHlqw}, we apply three reflections on the coin register if
  the vertex $x$ is marked, and we apply a single reflection if the vertex is
  unmarked. In~either case, we apply an odd number of reflections on the coin
  register. We can therefore rewrite the circuit in Figure~\ref{circForHlqw}
  as the equivalently acting circuit given in %
  Figure~\ref{circForHlqw simplified}.

  \begin{figure}[htbp]
    \centering \begin{tikzpicture}[scale=1.500000,x=1pt,y=1pt]
\filldraw[color=white] (0.000000, -8.500000) rectangle (130.000000, 20.500000);
% Drawing wires
% Line 3: a W breadth=40 size=40
\draw[color=black] (0.000000,28.500000) -- (130.000000,28.500000);
% Line 4: b W
\draw[color=black] (0.000000,0.000000) -- (130.000000,0.000000);
% Done with wires; drawing gates
% Line 5: a P $m$ b G:width=40 $\mathsf{Ref}(\ket{c^\perp})$
\draw (26.000000,28.500000) -- (26.000000,0.000000);
\begin{scope}
\draw[fill=white] (26.000000, 28.500000) circle(7.000000pt);
\clip (26.000000, 28.500000) circle(7.000000pt);
\draw (26.000000, 28.500000) node {$m$};
\end{scope}
\begin{scope}
\draw[fill=white] (26.000000, -0.000000) +(-45.000000:28.284271pt and 9.899495pt) -- +(45.000000:28.284271pt and 9.899495pt) -- +(135.000000:28.284271pt and 9.899495pt) -- +(225.000000:28.284271pt and 9.899495pt) -- cycle;
\clip (26.000000, -0.000000) +(-45.000000:28.284271pt and 9.899495pt) -- +(45.000000:28.284271pt and 9.899495pt) -- +(135.000000:28.284271pt and 9.899495pt) -- +(225.000000:28.284271pt and 9.899495pt) -- cycle;
\draw (26.000000, -0.000000) node {$\mathsf{Ref}(\ket{c^\perp})$};
\end{scope}
% Line 6: a P $u$ b G:width=40 $\mathsf{Ref}(\ket{c})$
\draw (78.000000,28.500000) -- (78.000000,0.000000);
\begin{scope}
\draw[fill=white] (78.000000, 28.500000) circle(7.000000pt);
\clip (78.000000, 28.500000) circle(7.000000pt);
\draw (78.000000, 28.500000) node {$u$};
\end{scope}
\begin{scope}
\draw[fill=white] (78.000000, -0.000000) +(-45.000000:28.284271pt and 9.899495pt) -- +(45.000000:28.284271pt and 9.899495pt) -- +(135.000000:28.284271pt and 9.899495pt) -- +(225.000000:28.284271pt and 9.899495pt) -- cycle;
\clip (78.000000, -0.000000) +(-45.000000:28.284271pt and 9.899495pt) -- +(45.000000:28.284271pt and 9.899495pt) -- +(135.000000:28.284271pt and 9.899495pt) -- +(225.000000:28.284271pt and 9.899495pt) -- cycle;
\draw (78.000000, -0.000000) node {$\mathsf{Ref}(\ket{c})$};
\end{scope}
% Line 7: a b G $\opshift$
\draw (117.000000,28.500000) -- (117.000000,0.000000);
\begin{scope}
\draw[fill=white] (117.000000, 14.250000) +(-45.000000:9.899495pt and 30.052038pt) -- +(45.000000:9.899495pt and 30.052038pt) -- +(135.000000:9.899495pt and 30.052038pt) -- +(225.000000:9.899495pt and 30.052038pt) -- cycle;
\clip (117.000000, 14.250000) +(-45.000000:9.899495pt and 30.052038pt) -- +(45.000000:9.899495pt and 30.052038pt) -- +(135.000000:9.899495pt and 30.052038pt) -- +(225.000000:9.899495pt and 30.052038pt) -- cycle;
\draw (117.000000, 14.250000) node {$\opshift$};
\end{scope}
% Done with gates; drawing ending labels
% Done with ending labels; drawing cut lines and comments
% Done with comments
\end{tikzpicture}
    \caption{Equivalent circuit of $\vlqw$.}
    \label{circForHlqw simplified}
  \end{figure}

  For the lazy random walk $\oplw$ with self-loops of weight~$\ell$, the
  neighborhood state of any vertex~$x$ on a regular graph is
  $\ket{\oplw_x} = \sqrt{\frac{d}{d+\ell}} \ket{\opp_x} +
  \sqrt{\frac{\ell}{d+\ell}} \ket{x}$.  The interpolation of a random walk
  changes the neighborhood state of any marked vertex.  With our choice of
  $s = 1 - \frac{\ell}{d}$, they become
  \begin{equation*}
    \begin{array}{ll}
      \ket{\liw_u} &= \sqrt{\frac{d}{d+\ell}} \ket{\opp_u} + \sqrt{\frac{\ell}{d+\ell}} \ket{u} \\
      \ket{\liw_m} &= \sqrt{\frac{\ell}{d+\ell}} \ket{\opp_m} + \sqrt{\frac{d}{d+\ell}} \ket{m}.
    \end{array}
  \end{equation*}
  Here $u$ denotes any unmarked vertex, and $m$ denotes the unique marked
  vertex. Applying the isometry~$\isoe$ yields that
  \begin{equation}\label{eq:isometry E on basis states}
    \begin{array}{ll}
      \isoe\, \ket{u, \liw_u} &= \ket{u, c} \\
      \isoe\, \ket{m, \liw_m} &= \ket{m, c^\perp}.
    \end{array}
  \end{equation}
  By definition, the $\opshift$ operator is equivalent to the $\opswap$
  operator under the isometry,
  \begin{equation}\label{shiftswapeq}
    \opshift = \isoe \cdot \opswap \cdot \isoe^\dagger.
  \end{equation}
  Eqs.~\ref{eq:isometry E on basis states} and~\ref{shiftswapeq} permit us to
  write the coined quantum walk circuit in Figure~\ref{circForHlqw simplified}
  as a circuit of the quantum interpolated
  walk~$\isoe \cdot \lqiw \cdot \isoe^\dagger$, as in
  Figure~\ref{circForLqiw}.
  \begin{figure}[htbp]
    \centering \begin{tikzpicture}[scale=1.500000,x=1pt,y=1pt]
\filldraw[color=white] (0.000000, -8.500000) rectangle (202.000000, 20.500000);
% Drawing wires
% Line 3: a W breadth=40 size=40
\draw[color=black] (0.000000,28.500000) -- (202.000000,28.500000);
% Line 4: b W
\draw[color=black] (0.000000,0.000000) -- (202.000000,0.000000);
% Done with wires; drawing gates
% Line 5: a b G $\isoe^\dagger$
\draw (13.000000,28.500000) -- (13.000000,0.000000);
\begin{scope}
\draw[fill=white] (13.000000, 14.250000) +(-45.000000:9.899495pt and 30.052038pt) -- +(45.000000:9.899495pt and 30.052038pt) -- +(135.000000:9.899495pt and 30.052038pt) -- +(225.000000:9.899495pt and 30.052038pt) -- cycle;
\clip (13.000000, 14.250000) +(-45.000000:9.899495pt and 30.052038pt) -- +(45.000000:9.899495pt and 30.052038pt) -- +(135.000000:9.899495pt and 30.052038pt) -- +(225.000000:9.899495pt and 30.052038pt) -- cycle;
\draw (13.000000, 14.250000) node {$\isoe^\dagger$};
\end{scope}
% Line 6: a P $m$ b G:width=50 $\mathsf{Ref}(\ket{\liw_m})$
\draw (57.000000,28.500000) -- (57.000000,0.000000);
\begin{scope}
\draw[fill=white] (57.000000, 28.500000) circle(7.000000pt);
\clip (57.000000, 28.500000) circle(7.000000pt);
\draw (57.000000, 28.500000) node {$m$};
\end{scope}
\begin{scope}
\draw[fill=white] (57.000000, -0.000000) +(-45.000000:35.355339pt and 9.899495pt) -- +(45.000000:35.355339pt and 9.899495pt) -- +(135.000000:35.355339pt and 9.899495pt) -- +(225.000000:35.355339pt and 9.899495pt) -- cycle;
\clip (57.000000, -0.000000) +(-45.000000:35.355339pt and 9.899495pt) -- +(45.000000:35.355339pt and 9.899495pt) -- +(135.000000:35.355339pt and 9.899495pt) -- +(225.000000:35.355339pt and 9.899495pt) -- cycle;
\draw (57.000000, -0.000000) node {$\mathsf{Ref}(\ket{\liw_m})$};
\end{scope}
% Line 7: a P $u$ b G:width=50 $\mathsf{Ref}(\ket{\liw_u})$
\draw (119.000000,28.500000) -- (119.000000,0.000000);
\begin{scope}
\draw[fill=white] (119.000000, 28.500000) circle(7.000000pt);
\clip (119.000000, 28.500000) circle(7.000000pt);
\draw (119.000000, 28.500000) node {$u$};
\end{scope}
\begin{scope}
\draw[fill=white] (119.000000, -0.000000) +(-45.000000:35.355339pt and 9.899495pt) -- +(45.000000:35.355339pt and 9.899495pt) -- +(135.000000:35.355339pt and 9.899495pt) -- +(225.000000:35.355339pt and 9.899495pt) -- cycle;
\clip (119.000000, -0.000000) +(-45.000000:35.355339pt and 9.899495pt) -- +(45.000000:35.355339pt and 9.899495pt) -- +(135.000000:35.355339pt and 9.899495pt) -- +(225.000000:35.355339pt and 9.899495pt) -- cycle;
\draw (119.000000, -0.000000) node {$\mathsf{Ref}(\ket{\liw_u})$};
\end{scope}
% Line 8: a b G \rotatebox{90}{$\opswap$}
\draw (163.000000,28.500000) -- (163.000000,0.000000);
\begin{scope}
\draw[fill=white] (163.000000, 14.250000) +(-45.000000:9.899495pt and 30.052038pt) -- +(45.000000:9.899495pt and 30.052038pt) -- +(135.000000:9.899495pt and 30.052038pt) -- +(225.000000:9.899495pt and 30.052038pt) -- cycle;
\clip (163.000000, 14.250000) +(-45.000000:9.899495pt and 30.052038pt) -- +(45.000000:9.899495pt and 30.052038pt) -- +(135.000000:9.899495pt and 30.052038pt) -- +(225.000000:9.899495pt and 30.052038pt) -- cycle;
\draw (163.000000, 14.250000) node {\rotatebox{90}{$\opswap$}};
\end{scope}
% Line 9: a b G $\isoe$
\draw (189.000000,28.500000) -- (189.000000,0.000000);
\begin{scope}
\draw[fill=white] (189.000000, 14.250000) +(-45.000000:9.899495pt and 30.052038pt) -- +(45.000000:9.899495pt and 30.052038pt) -- +(135.000000:9.899495pt and 30.052038pt) -- +(225.000000:9.899495pt and 30.052038pt) -- cycle;
\clip (189.000000, 14.250000) +(-45.000000:9.899495pt and 30.052038pt) -- +(45.000000:9.899495pt and 30.052038pt) -- +(135.000000:9.899495pt and 30.052038pt) -- +(225.000000:9.899495pt and 30.052038pt) -- cycle;
\draw (189.000000, 14.250000) node {$\isoe$};
\end{scope}
% Done with gates; drawing ending labels
% Done with ending labels; drawing cut lines and comments
% Done with comments
\end{tikzpicture}
    \caption{Circuit of $\vlqw$ written as a circuit of the quantum
      interpolated walk $\isoe \cdot \lqiw \cdot \isoe^\dagger$.}
    \label{circForLqiw}
  \end{figure}
\end{proof}

We remark that the value of $\ell$ used in~\cite{RW20} is
$\ell = \frac{d}{N}$, and that the value of $\ell$ proposed
in~\cite{WZWY17arxiv} is $\ell = \frac{d}{N-1}$.  The difference between these
two values of $\ell$ is an additive term of order~$\frac{1}{N^2}$.  In this
work, we pick the value of $\ell$ to be equal to $\frac{d}{N}$, so that the
correspondence with~\cite{RW20} in Lemma~\ref{lem2} is exact.  The value of
$s$ used in~\cite{KMOR16}, and in the simulation in Section~7 in~\cite{DH17b},
is $1-\frac{1}{N-1}$, which corresponds to a value of $\ell$ equal to
$\frac{d}{N-1}$.  By Eqs.~(192) and~(21) in~\cite{KMOR16}, our slightly
different choice of $\ell$ implies that the $\intht{\iw}$ used in our paper is
a factor of order $\frac{1}{N}$ larger than the $\intht{\iw}$ used
in~\cite{KMOR16}.  This negligible factor does not change the results stated
in this paper.

By Lemma~\ref{lem1}, $\cotqht{\lqw}{\initlazy} =
\cotqht{\vlqw}{\initlazy}$. By Lemma~\ref{lem2} and the definition of the
isometry $\isoe$, $\cotqht{\vlqw}{\initlazy} = \cotqht{\lqiw}{\initlip}$. We
next show the exact relationship between $\cotqht{\lqiw}{\initlip}$ and
$\cotqht{\qiw}{\initip}$.

\begin{lemma}\label{lem3}
  For $\ell = \frac{d}{N}$ and $s = 1 - \frac{\ell}{d} = 1 - \frac{1}{N}$,
  \begin{equation*}
    \cotqht{\lqiw}{\initlip}^2 = \frac{N+1}{N}\cotqht{\qiw}{\initip}^2 + \frac{1}{2N-1}.
  \end{equation*}
\end{lemma}

\begin{proof}
  By definitions of $\iw$ and $\liw$, given by Eqs.~\ref{lazy walk}
  and~\ref{interpolated lazy walk}, respectively,
  \begin{equation*}
    \liw = \frac{N}{N+1} \cdot \iw + \frac{1}{N+1} \cdot \opid_N.
  \end{equation*}
  Hence $\liw$ and $\iw$ have the same eigenvectors
  $\ket{\widehat\lambda_k} = \ket{\lambda_k}$ and corresponding eigenvalues
  \begin{equation}\label{equation of eval}
    \widehat\lambda_k = \frac{N}{N+1} \lambda_k + \frac{1}{N+1}.
  \end{equation}
  Then by the definition of the interpolated hitting time in Eq.~(\ref{ip
    hitting time}),
  \begin{equation}\label{relationship of ip ht}
    \intht{\liw} = \frac{N+1}{N} \cdot \intht{\iw}.
  \end{equation}
  Given an interpolated walk $\iw$ and its Szegedy's correspondence $\qiw$, we
  have
  \begin{equation}\label{ip ht and cot qht}
    \cotqht{\qiw}{\initip}^2 = 2 \intht{\iw} - \frac{p_M}{1-s(1-p_M)},
  \end{equation}
  by direct calculation using Definition~\ref{cotangent qht} and Eq.~(\ref{ip
    hitting time}).  Here $p_M$ is the probability of drawing a marked vertex
  from the stationary distribution~$\uppi$.  Since our graph is regular and
  there is a unique marked vertex, $p_M = \uppi_m = \frac{1}{N}$.
  Lemma~\ref{lem3} follows from plugging Eq.~(\ref{ip ht and cot qht}) into
  Eq.~(\ref{relationship of ip ht}) on both sides.
\end{proof}

Theorem~\ref{thm1} now follows from Lemmas~\ref{lem1},~\ref{lem2}
and~\ref{lem3}. The fact that $\cotqht{\qiw}{\initip}$ is in
$O\bigl(\sqrt{\sht(\opp, \{m\})}\bigr)$ follows from~\cite{KMOR16}
and~\cite{DH17b}.

\section{Proof of Theorem~\ref{thm2}}\label{proof of thm2}

Lemmas~\ref{lem1} and~\ref{lem2} in the previous section give exact
relationships between the lackadaisical quantum walk operator $\lqw$ and the
two intermediate walk operators $\vlqw$ and~$\lqiw$.  In this section, we then
show that the distance between the intermediate walk operator $\lqiw$ and the
quantum interpolated walk operator $\qiw$ is bounded, when applied on their
respective initial states.  The bound is given in Lemma~\ref{lem6} below, and
it follows from Lemmas~\ref{lem4} and~\ref{lem5}.  Theorem~\ref{thm2} follows from
Lemmas~\ref{lem2} and~\ref{lem6}.

By Szegedy's correspondence, the eigenvectors $\ket{\phi_k^\pm}$ of $\qiw$ are
in the subspace
\begin{equation*}
  \text{span}\big\{\isops\ket{\lambda_k}, \opswap\cdot \isops\ket{\lambda_k} \big\} = \text{span}\big\{\isops\ket{\lambda_k}, (\isops\ket{\lambda_k})^\perp\big\},
\end{equation*}
and the eigenvectors $\lazyphi{k}{\pm}$ of $\lqiw$ are in the subspace
\begin{equation*}
  \text{span}\big\{\isolazy\ket{\lambda_k}, \opswap\cdot \isolazy\ket{\lambda_k} \big\} = \text{span}\big\{\isolazy\ket{\lambda_k}, (\isolazy\ket{\lambda_k})^\perp\big\}.
\end{equation*}
We define an isometry
\begin{equation*}
  \isoone = \sum_x \big(\ketbra{x, \liw_x}{x, \iw_x} +  \ketbra{x, \liw_x^\perp}{x, \iw_x^\perp}\big),
\end{equation*}
where $\ket{x, \iw_x^\perp}$ is orthogonal to $\ket{x, \iw_x}$ in the subspace
spanned by $\{ \ket{x, \iw_x}, \ket{\iw_x ,x}\}$ and $\ket{x, \liw_x^\perp}$
is orthogonal to $\ket{x, \liw_x}$ in the subspace spanned by
$\{ \ket{x, \liw_x}, \ket{\liw_x, x}\}$. The isometry $\isoone$ satisfies that
\begin{equation*}
  \isoone\colon
  \left\{
    \begin{aligned}
      &\ket{\phi_k^+} &&\mapsto&& \lazyphi{k}{+} \\
      &\ket{\phi_k^-} &&\mapsto&& \lazyphi{k}{-} \\
      &\ket{\phi_{\nn}} &&\mapsto&& \lazyphi{\nn}{}.
    \end{aligned}
  \right.
\end{equation*}
By Eq.~(\ref{decompose state}), applying $\isoone$ on the state $\qiw^t \initip$
changes from the eigenspace of $\qiw$ to the eigenspace of $\lqiw$,
\begin{equation*}
  \isoone \cdot \qiw^t \initip = \alpha_{\nn} {\lazyphi{n}{}} +  \frac{1}{\sqrt{2}} \sum_{k=1}^{\nn-1} \alpha_k \bigl((e^{i\theta_k})^t {\lazyphi{k}{+}} + (e^{-i\theta_k})^t {\lazyphi{k}{-}}\bigr).
\end{equation*}
In the proof of Lemma~\ref{lem4} below, we require a second isometry, which is
also a projection,
\begin{equation*}
  \isotwo = \sum_x \bigl( \ketbra{x, \iw_x}{x, \iw_x} + \ketbra{x, \iw_x^\perp}{x, \iw_x^\perp} \bigr).
\end{equation*}
Applying $\isotwo$ on $\qiw^t \initip$ does not change the state itself,
\begin{equation*}
  \isotwo \cdot \qiw^t \initip= \qiw^t \initip.
\end{equation*}
Note that since the states $\ket{\iw_x}$ and $\ket{\liw_x}$ are close for all
vertices $x$, the $\ell_2$-distance between $\isoone$ and $\isotwo$ is small,
which is $O\bigl(\frac{1}{\sqrt{N}}\bigr)$ by direct calculation.

\begin{lemma}\label{lem4}
  For all $t \geq 0$,
  \begin{equation*}
    \dnorm[\big]{\isoone \cdot \qiw^t \cdot \initip -  \qiw^t \cdot \initip}_2 \in O\biggl(\frac{1}{\sqrt{N}}\biggr).
  \end{equation*}
\end{lemma}

\begin{proof}
  \begin{align*} 
    & \hphantom{{}={}}
      \dnorm[\big]{\isoone \cdot \qiw^t \cdot \initip -  \qiw^t \cdot \initip}_2\\
    &= \dnorm[\big]{{\isoone \cdot \qiw^t \cdot \initip -  \isotwo \cdot \qiw^t \cdot \initip}}_2 \\
    &\leq \dnorm[\big]{\isoone - \isotwo}_2 \cdot \dnorm[\big]{\qiw^t \cdot \initip}_2
    %\\
    %&
       = \dnorm[\big]{\isoone - \isotwo}_2
    %\\
    %&
      \in O\biggl(\frac{1}{\sqrt{N}}\biggr).
  \end{align*}
\end{proof}

Next consider the following $\ell_2$-distance
\begin{align*}
  & \hphantom{{}={}}
    \dnorm[\big]{ \isoone \cdot  \qiw^t \cdot \initip - \lqiw^t \cdot \initlip }_2 \\
  & = \dnorm[\bigg]{\frac{1}{\sqrt{2}} \sum_{k=1}^{\nn-1} \alpha_k \Bigl((e^{ti\theta_k} - e^{ti\widehat\theta_k}){\lazyphi{k}{+}} + (e^{-ti\theta_k} - e^{-ti\widehat\theta_k}){\lazyphi{k}{-}}\Bigr) }_2.
\end{align*}
We pick a threshold angle $\theta_0$ satisfying that
$0 < \theta_0 \leq \frac{\pi}{2}$ and separate the sum into two parts, where
the first part is for angles $0 < \theta_k \leq \theta_0$ and the second part
is for $\theta_0 < \theta_k \leq \pi$.  We give an upper bound on the
$\ell_2$-norm for each of these two parts in Facts~\ref{factsmall}
and~\ref{factlarge}, respectively.

\begin{fact}
  \label{factsmall}
  For $0 < \theta_0 \leq \frac{\pi}{2}$ and all $t \geq 0$,
  \begin{equation*}
    \dnorm[\Bigg]{\frac{1}{\sqrt{2}} \sum_{0 < \theta_k \leq \theta_0}
      \alpha_k \Bigl((e^{ti\theta_k} - e^{ti\widehat\theta_k})\lazyphi{k}{+} +
      (e^{-ti\theta_k} - e^{-ti\widehat\theta_k})\lazyphi{k}{-}\Bigr)}_2 \leq
    \frac{8t}{N-1} \sin\biggl(\frac{\theta_0}{2}\biggr).
  \end{equation*}
\end{fact}

\begin{proof}
  \begin{align*}
    & \hphantom{{}={}}
      \dnorm[\Bigg]{\frac{1}{\sqrt{2}} \sum_{0 < \theta_k \leq \theta_0} \alpha_k \Bigl((e^{ti\theta_k} - e^{ti\widehat\theta_k})\lazyphi{k}{+} + (e^{-ti\theta_k} - e^{-ti\widehat\theta_k})\lazyphi{k}{-}\Bigr) }_2 \\
    & = \frac{1}{\sqrt{2}} \sqrt{ \sum_{0 < \theta_k \leq \theta_0} \dabs{\alpha_k}^2 \Bigl(\dabs{e^{ti\theta_k} - e^{ti\widehat\theta_k}}^2  + \dabs{e^{-ti\theta_k} - e^{-ti\widehat\theta_k}}^2  \Bigr)}\\
    & \leq \max_{\theta_k \leq \theta_0} \,\dabs[\big]{e^{ti\theta_k} - e^{ti\widehat\theta_k}} \sqrt{\sum_{0 < \theta_k \leq \theta_0} \dabs{\alpha_k}^2}\\
    & \leq \max_{\theta_k \leq \theta_0}\, \dabs[\big]{e^{ti\theta_k} - e^{ti\widehat\theta_k}}\\
    & = \max_{\theta_k \leq \theta_0} \,\dabs[\bigg]{
      2\sin\biggl(t\frac{\theta_k - \widehat\theta_k}{2}\biggr)}\\
    & = \max_{\theta_k \leq \theta_0} \,\dabs[\bigg]{
      2\sin\biggl(\frac{2t}{N+1}\frac{N+1}{2}\frac{\theta_k - \widehat\theta_k}{2}\biggr)}\\
    & \leq \frac{8t}{N-1}
      \sin\bigg(\frac{\theta_0}{2}\bigg). 
  \end{align*}
  In the last inequality, by Eq.~(\ref{equation of eval}), since
  $0 < \theta_k \leq \theta_0 \leq \frac{\pi}{2}$, then
  $(1-\frac{2}{N+1})\theta_k \leq \widehat\theta_k \leq \theta_k$, which
  implies that
  $0 \leq \frac{N+1}{2} \frac{\theta_k - \widehat\theta_k}{2} \leq
  \frac{\theta_0}{2}$.  Finally $\sin(ax) \leq 2a \sin(x)$ for all
  $0 \leq x \leq \pi/4$ and all $a \geq 0$.
\end{proof}

\begin{fact}
  \label{factlarge}
  For $0 < \theta_0 \leq \frac{\pi}{2}$ and all $t \geq 0$,
  \begin{equation*}
    \dnorm[\Bigg]{\frac{1}{\sqrt{2}} \sum_{\theta_0 < \theta_k \leq \pi} 
      \alpha_k \Bigl((e^{ti\theta_k} - e^{ti\widehat\theta_k})\lazyphi{k}{+} 
      + (e^{-ti\theta_k} - e^{-ti\widehat\theta_k})\lazyphi{k}{-}\Bigr)}_2 
    \leq \frac{2}{\sqrt{(1-\cos\theta_0)(N-1)}}.
  \end{equation*}
\end{fact}

\begin{proof}
  First simplify by bounding each factor
  $(e^{-ti\theta_k} - e^{-ti\widehat\theta_k})$ by its trivial upper bound
  of~2.
  \begin{align*}
    & \hphantom{{}={}}
      \dnorm[\Bigg]{\frac{1}{\sqrt{2}} \sum_{\theta_0 < \theta_k \leq \pi} \alpha_k \Bigl((e^{ti\theta_k} - e^{ti\widehat\theta_k})\lazyphi{k}{+} + (e^{-ti\theta_k} - e^{-ti\widehat\theta_k})\lazyphi{k}{-}\Bigr)}_2 \\
    & \leq \max_{\theta_0 < \theta_k} \,\dabs[\big]{e^{ti\theta_k} -
      e^{ti\widehat\theta_k}} \sqrt{\sum_{\theta_0 < \theta_k \leq \pi}
      \dabs{\alpha_k}^2} 
     \leq 2 \sqrt{\sum_{\theta_0 < \theta_k \leq \pi} \dabs{\alpha_k}^2}.
  \end{align*}
  We next upper bound the sum of the scalars $\alpha_k^2$ for the large angles
  $\theta_0 < \theta_k \leq \pi$ by $ \frac{1}{(1-\cos\theta_0)(N-1)}$.  For
  this, consider the quantum walk $\qiw$,
  \begin{equation*}
    \bra{\overline{init}_\mathrm{ip}} \qiw \initip 
    = \sqrt{\bar\uppi}^\dagger \cdot \dm{\iw} \cdot \sqrt{\bar\uppi} 
    = 1 - \frac{1}{N - 1}.
  \end{equation*}
  By Eq.~(\ref{decompose pibar}), we have
  $\sqrt{\bar\uppi}^\dagger \cdot \dm{\iw} \cdot \sqrt{\bar\uppi} =
  \sum_{k=1}^{\nn} \lambda_k \alpha_k^2$.  Using that
  $\sum_{k=1}^{\nn} \alpha_k^2 = 1$, we infer that
  \begin{align*}
    1 - \sqrt{\bar\uppi}^\dagger \cdot \dm{\iw} \cdot \sqrt{\bar\uppi} 
    & =  \sum_{k=1}^{\nn} \alpha_k^2  - \sum_{k=1}^{\nn} \lambda_k \alpha_k^2 = \sum_{k=1}^{\nn} (1-\lambda_k) \alpha_k^2\\
    & = \sum_{0 < \theta_k \leq \theta_0}(1-\lambda_k) \alpha_k^2 + \sum_{\theta_0 < \theta_k \leq \pi}(1-\lambda_k) \alpha_k^2 
      = \frac{1}{N - 1}.
  \end{align*}
  Since $-1 \leq \lambda_k \leq 1$ for all $k$, we have
  $\sum_{0 < \theta_k \leq \theta_0}(1-\lambda_k) \alpha_k^2 \geq 0$. For
  $\theta_0 < \theta_k \leq \pi$,\ i.e.\ $\lambda_k < \cos\theta_0$, we
  conclude that
  \begin{equation*}
    \sum_{\theta_0 < \theta_k \leq \pi}\alpha_k^2 
    \leq \frac{1}{1-\cos\theta_0} 
    \sum_{\theta_0 < \theta_k \leq \pi}(1-\lambda_k) \alpha_k^2 
    \leq \frac{1}{(1-\cos\theta_0)(N-1)}.
  \end{equation*}
\end{proof}

\begin{lemma}\label{lem5}
  Fix a constant $c\geq 1$, then for all $t \leq c\sqrt{\sht(\opp, \{m\})}$,
  \begin{equation*}
    \dnorm[\big]{ \isoone \cdot  \qiw^t \cdot \initip - \lqiw^t \cdot \initlip }_2 
    \in O\biggl(\frac{1}{N^{1/4}}\biggr).
  \end{equation*}
\end{lemma}

\begin{proof}
  Choose the threshold angle $\theta_0$ such that
  $\cos\theta_0 = 1 - 2 \sqrt{\frac{N-1}{16\sht(\opp, \{m\})}}$.  Since the
  hitting time $\sht(\opp, \{m\})$ for a connected regular graph is at least
  $N-1$, the threshold angle is well-defined and satisfies that
  $0< \theta_0 \leq \pi/2$, and thus $0 \leq \cos\theta_0 < 1$.  

  Apply the triangle inequality on Facts~\ref{factsmall} and~\ref{factlarge}, and
  substitute
  $\sin\bigl(\frac{\theta_0}{2}\bigr) = \Bigl(\frac{N-1}{16\sht(\opp,
    \{m\})}\Bigr)^{{1/4}}$.
  \begin{align*}
    & \hphantom{{}={}}
    \dnorm[\big]{ \isoone \cdot  \qiw^t \cdot \initip - \lqiw^t \cdot \initlip }_2\\
    &=\dnorm[\Bigg]{\frac{1}{\sqrt{2}} 
      \sum_{0 < \theta_k \leq \pi} \alpha_k \Bigl((e^{ti\theta_k} - e^{ti\widehat\theta_k})\lazyphi{k}{+} 
      + (e^{-ti\theta_k} - e^{-ti\widehat\theta_k})\lazyphi{k}{-}\Bigr)}_2 \\
    & \leq \frac{8t}{N-1}\sin\biggl(\frac{\theta_0}{2}\biggr) 
      + \frac{2}{\sqrt{(1-\cos\theta_0)(N-1)}}\\
    & = \frac{8t}{N-1} \biggl(\frac{N-1}{16\sht(\opp, \{m\})}\biggr)^{{1/4}} 
      +  \frac{\sqrt{2}}{\sqrt{N-1}}\biggl(\frac{16\sht(\opp, \{m\})}{N-1}\biggr)^{{1/4}}\\
    & \leq \frac{1}{(N-1)^{3/4}} \Biggl(\frac{8t}{\bigl(16\sht(\opp, \{m\})\bigr)^{1/4}} 
      + \sqrt{2} \bigl(16\sht(\opp, \{m\})\bigr)^{1/4} \Biggr)\\
    &\leq \frac{\sht(\opp, \{m\})^{1/4}}{(N-1)^{3/4}} \bigl(4c + 2\sqrt{2}\bigr)\\
    &\leq \frac{2^{1/4} \sqrt{N}}{(N-1)^{3/4}} \bigl(4c + 2\sqrt{2}\bigr)\\
    &\leq 9c\frac{\sqrt{N}}{(N-1)^{3/4}}\\
    &\in O\biggl(\frac{1}{N^{{1/4}}}\biggr).
  \end{align*}
  In the second last inequality, we apply the upper bound on the hitting time
  of random walks on regular graphs given in~\cite{Fei96}, which shows that
  $\sht(\opp, \{m\}) \leq 2N^2$.
\end{proof}

\begin{lemma}\label{lem6}
  Fix a constant $c\geq 1$, then for all $t \leq c\sqrt{\sht(\opp, \{m\})}$,
  \begin{equation*}
    \dnorm[\big]{\qiw^t \cdot \initip - \lqiw^t \cdot \initlip }_2 
    \in O\biggl(\frac{1}{N^{1/4}}\biggr).
  \end{equation*}
\end{lemma}

\begin{proof}
  Applying the triangle inequality on Lemmas~\ref{lem4} and~\ref{lem5}.
\end{proof}

Theorem~\ref{thm2} follows from Lemmas~\ref{lem2} and~\ref{lem6}.

\section{On Locally Arc-Transitivity}
\label{sec_arc}

The main property that we have used in our proofs is locally arc-transitivity.
The graphs considered in~\cite{Won15b,Won18b,WZWY17arxiv,GK19,RW19b,RW20} are
all locally arc-transitive, as well as vertex-transitive.  This implies that
these graphs are also symmetric, as we now show.

A graph is \emph{symmetric} if for any two arcs $(u_1,v_1)$ and $(u_2,v_2)$,
there exists an automorphism that maps $(u_1,v_1)$ to $(u_2,v_2)$.  As
discussed in Sec.~\ref{sec_main}, a graph $G$ is locally arc-transitive if for
any vertex $u$ with neighbors $v_1$ and $v_2$, there exists an automorphism of
$G$ that maps the arc $(u,v_1)$ to the arc $(u,v_2)$.  Let us say that a graph
$G$ is \emph{locally arc-transitive at vertex $u$} if for any two neighbors
$v_1$ and $v_2$ of $u$, there exists an automorphism of $G$ that maps the arc
$(u,v_1)$ to the arc $(u,v_2)$.  A locally arc-transitive graph is then, by
definition, a graph that is locally arc-transitive at every vertex.

We now show that a graph $G$ is symmetric if and only if it is locally
arc-transitive at some vertex $u$ and vertex-transitive.  Trivially, if a
graph is symmetric, it satisfies the latter conditions.  Now, consider the
converse.  Let $(x_1,y_1)$ and $(x_2,y_2)$ be two arcs in~$G$.  Using
vertex-transitivity, let $\sigma_1$ be an automorphism that maps $x_1$ to $u$,
and let $\sigma_2$ be an automorphism that maps $x_2$ to $u$.  Using locally 
arc-transitivity at~$u$, let $\sigma_3$ be an automorphism that maps $u$ to
$u$, and that maps $\sigma_1(y_1)$ to $\sigma_2(y_2)$.  Then
$\sigma_2^{-1} \sigma_3 \sigma_1$ maps $(x_1,y_1)$ to $(x_2,y_2)$, and thus
$G$ is symmetric.

Locally arc-transitivity and vertex-transitivity are two distinct graph
properties, even when restricting to regular graphs.  There are regular graphs
that are locally arc-transitive but not vertex-transitive, such as the Folkman
graph~\cite{Fol67}.  Conversely, there are regular graphs that are
vertex-transitive but not locally arc-transitive, such as the M\"obius ladder
on 8 vertices, or the Cartesian product of a 3-cycle and a 2-path.

There are two obvious cases of graph properties that are not included in our
main theorems.  The first case is graphs that are locally arc-transitive but
not regular, which includes the bipartite graphs considered in~\cite{RW19b}.
In this case, if we equip each vertex with a self-loop of weight proportional
to its degree, the equivalence to quantum interpolated walks still holds,
though the corresponding interpolated walk~$\iw$ may use a value of~$s$
different from the value $s = 1 - \frac{\uppi_m}{1-\uppi_m}$ picked in a
quantum interpolated walk, depending on the proportionality factor.

The second case is graphs that are regular and vertex-transitive, but not
locally arc-transitive.  In this case, our theorems do not apply, but a
partial analytical answer may be obtained by bounding the variation between
the amplitudes of the neighbors of the marked vertex after each iteration of
the walk.

\section{Conclusion}
\label{sec_conclusion}

We analytically prove that lackadaisical quantum walks can find a unique
marked vertex on any regular locally arc-transitive graph with constant
success probability in $O(\sqrt{\sht})$ steps.  In our proof, we establish and
use relationships between lackadaisical quantum walks and quantum interpolated
walks for any regular locally arc-transitive graph.  We also prove that self-loops of
weight $\ell$ correspond to an interpolation parameter of
$s = 1- \frac{\ell}{d}$.

Our results prove several speculations and numerical findings in previous
work, including the conjectures that lackadaisical quantum walks can find a
unique marked vertex with constant success probability on the
torus~\cite{Won18b, WZWY17arxiv}, the cycle~\cite{GK19}, Paley graphs, some
Latin square graphs, Johnson graphs, and the hypercube~\cite{RW20}.

\section*{Acknowledgments}
The authors are grateful to Dante Bencivenga, Xining Chen, Janet Leahy and
Shang Li for discussions. We are grateful to Mason Rhodes and Thomas Wong for
discussions on the results in their work~\cite{RW20}, and to an anonymous
referee for helpful comments.  This work was supported in part by the Alberta
Graduate Excellence Scholarship program (AGES), the Natural Sciences and
Engineering Research Council of Canada (NSERC), and the University of
Calgary's Program for Undergraduate Research Experience (PURE).

\end{document}